\documentclass[sii]{ipart}

\RequirePackage{hyperref,cellspace,booktabs,graphicx,natbib}

\startlocaldefs
\newtheorem{lemma}{Lemma}
\newtheorem{definition}{Definition}
\newtheorem{remark}{Remark}
\newtheorem{th1}{Theorem}
\newtheorem{col}{Corollary}
\newcommand{\Var}{\mathrm{Var}}

\endlocaldefs

\pubyear{2021}
\volume{0}
\issue{0}
\firstpage{1}
\lastpage{1}
\begin{document}

\begin{frontmatter}

%%  "Title of the Paper"
\title[Metric Distributional Discrepancy]{Metric Distributional Discrepancy in Metric Space}

\begin{aug}
    \author{\inits{W.}\fnms{Wenliang}
 \snm{Pan}},\thanksref{t1}
    \address{Faculty of Information Technology, Macau University of Science and Technology, Macao, China\\
    School of Mathematics, Sun Yat-sen University, Guangzhou,China }
    \author{\inits{Y.}\fnms{Yujue} \snm{Li}},\thanksref{t1}\thankstext{t1}{These authors contributed equally to this work}
    \address{School of Mathematics, Sun Yat-sen University, Guangzhou,China }
    \author{\inits{J.}\fnms{Jianwu} \snm{Liu}},
    \address{School of Mathematics, Sun Yat-sen University, Guangzhou,China}
    \author{\inits{P.}\fnms{Pei} \snm{Dang}}
    \address{Faculty of Information Technology, Macau University of Science and Technology, Macao, China}
    \and
    \author{\inits{W.}\fnms{Weixiong} \snm{Mai};\thanksref{t2}
    \ead[label=e3]{wxmai@must.edu.mo}}
    \address{Macao Center for Mathematical Sciences, Macau University of Science and Technology, Macao, China\\
    \printead{e3}}
    \thankstext{t2}{Corresponding author.}
    \author{\inits{t.} for the Alzheimer’s Disease Neuroimaging Initiative}
\end{aug}
%
%%  History:
%\received{\sday{3} \smonth{1} \syear{2019}}

\begin{abstract}
Independence analysis is an indispensable step before regression analysis to find out the essential factors that influence the objects. With many applications in machine Learning, medical Learning and a variety of disciplines, statistical methods of measuring the relationship between random variables have been well studied in vector spaces. However, there are few methods developed to verify the relation between random elements in metric spaces. In this paper, we present a novel index called metric distributional discrepancy (MDD) to measure the dependence between a random element $X$ and a categorical variable $Y$, which is applicable to the medical image and related variables. The metric distributional discrepancy statistics can be considered as the distance between the conditional distribution of $X$ given each class of $Y$ and the unconditional distribution of $X$. MDD enjoys some significant merits compared to other dependence-measures. For instance, MDD is zero if and only if $X$ and $Y$ are independent. MDD test is a distribution-free test since there is no assumption on the distribution of random elements. Furthermore, MDD test is robust to the data with heavy-tailed distribution and potential outliers. We demonstrate the validity of our theory and the property of the MDD test by several numerical experiments and real data analysis.
\end{abstract}

%\begin{keyword}[class=AMS]
%\kwd[Primary ]{}
%\kwd{}
%\kwd[; secondary ]{}
%\end{keyword}

%%  Upper case for every keyword

%\tableofcontents

\end{frontmatter}

%%  The body
\section{Introduction}
In view of the ever-growing and complex data in today's scientific world, there is an increasing need for a generic method to deal with these datasets in diverse application scenarios. Non-Euclidean data, such as brain imaging, computational biology, computer graphics and computational social sciences and among many others \cite{kong, Kendall, bookstein1996biometrics, muller, balasundaram2011clique, overview}, arises in many domains. For instance, images and time-varying data can be presented as functional data that are in the form of functions \cite{csenturk2011varying, csenturk2010functional, horvath2012inference}. The sphere, Matrix groups, Positive-Definite Tensors and shape spaces are also included as manifold examples \cite{rahman2005multiscale, shi2009intrinsic, lin2017extrinsic}. It is of great interest to discover the associations in those data. Nevertheless, many traditional analysis methods that cope with data in Euclidean spaces become invalid since non-Euclidean spaces are inherently nonlinear space without inner product. The analysis of these Non-Euclidean data presents many mathematical and computational challenges.

One major goal of statistical analysis is to understand the relationship among random vectors, such as measuring a linear/ nonlinear association between data, which is also a fundamental step for further statistical analysis (e.g., regression modeling). Correlation is viewed as a technique for measuring associations between random variables. A variety of classical methods has been developed to detect the correlation between data. For instance, Pearson's correlation and canonical correlation analysis (CCA) are powerful tools for capturing the degree of linear association between two sets of multi-variate random variables \cite{benesty2009pearson, fukumizu2005consistency, kim2014canonical}. 
% Linear regression can be conducted to measure the linear relationship between a predictor and an outcome variable \cite{zou2003correlation}. 
In contrast to Pearson's correlation, Spearman's rank correlation coefficient, as a non-parametric measure of rank correlation, can be applied in non-linear conditions \cite{sedgwick2014spearman}. 
% There is a polynomial regression that can fit a nonlinear relationship between variables but it is too sensitive to the outliers \cite{klivans2018efficient}.

However, statistical methods for measuring the association between complex structures of Non-Euclidean data have not been fully accounted in the methods above. In metric space, \cite{frechet} proposed a generalized mean in metric spaces and a corresponding variance that may be used to quantify the spread of the distribution of metric spaces valued random elements. However, in order to guarantee the existence and uniqueness of Fr\'{e}chet mean, it requires the space should be with negative sectional curvature. While in a positive sectional curvature space, the extra conditions such as bound on the support and radial distribution are required \cite{charlier2013necessary}. Following this, more nonparametric methods for manifold-valued data was developed. For instance, A Riemannian CCA model was proposed by \cite{kim2014canonical}, measuring an intrinsically linear association between two manifold-valued objects. Recently, \cite{articleSze} proposed distance correlation to measure the association between random vectors. After that, \cite{article2009} introduced Brownian covariance and showed it to be the same as distance covariance and \cite{lyons2013distance} extended the 
distance covariance to metric space under the condition that the space should be of strong negative type. Pan et.al \cite{pan2018ball,pan2019ball} introduced the notions of ball divergence and ball covariance for Banach-valued random vectors. These two notions can also be extended to metric spaces but by a less direct approach. Note that a metric space is endowed with a distance, it is worth studying the behaviors of distance-based statistical procedures in metric spaces.

In this paper, we extend the method in \citep{cui2018a} and propose a novel statistics in metric spaces based on \cite{wang2021nonparametric}, called metric distributional discrepancy (MDD), which considers a closed ball with the defined center and radius. We perform a powerful independence test which is applicable between a random vector $X$ and a categorical variable $Y$ based on MDD. The MDD statistics can be regarded as the weighted average of Cram\'{e}r-von Mises distance between the conditional distribution of $X$ given each class of $Y$ and the unconditional distribution of $X$. Our proposed method has the following major advantages, (i) $X$ is a random element in metric spaces. (ii) MDD is zero if and only if $X$ and $Y$ are statistically independent. (iii) It works well when the data is in heavy-tailed distribution or extreme value since it does not require any moment assumption. {Compared to distance correlation, MDD as an alternative dependence measure can be applied in the metric space which is not of strong negative type. Unlike ball correlation, MDD gets rid of unnecessary ball set calculation of $Y$ and has higher test power, which is shown in the simulation and real data analysis. }

The organization of the rest of this paper is as follows. In section \ref{secsta}, we give the definition and theoretical properties of MDD, and present results of monte carlo simulations in section \ref{secexp} and experiments on two real data analysis in section \ref{realdata}. Finally, we draw a conclusion in section \ref{conclu}.

\section{Metric distributional discrepancy}
\label{secsta}
\subsection{Metric distributional discrepancy statistics}
For convenience, we first list some notations in the metric space. The order pair $(\mathcal{M},d)$ denotes a $metric$ $space$ if $\mathcal{M}$ is a set and $d$ is a $metric$ or $distance$ on $\mathcal{M}$. Given a metric space $(\mathcal{M},d)$, let $\bar{B}(x, y) = \{ v : d(x, v) \leq r\}$ be a closed ball with the center $x$ and the radius $ r = d(x, y)$.  

Next, we define and illustrate the metric distributional discrepancy (MDD) statistics for a random element and a categorical one. Let $X$ be a random element and $Y$ be a categorical random variable with $R$ classes where $Y = \{y_1,y_2, \ldots, y_R\}$. Then, we let $X''$ be a copy of random element $X$, $F(x, x') = P_{X''}\{X''\in\bar{B}(x, x')\}$ be the unconditional distribution function of $X$, and $F_r(x, x')= P_{X''}\{X''\in\bar{B}(x, x')|Y = y_r\}$ be the conditional distribution function of $X$ given $Y = y_r$. The $MDD(X|Y)$ can be represented as the following quadratic form between $F(x, x')$ and $F_r(x, x')$,
 
\begin{equation}
\label{mvStat}
MDD(X|Y) = \sum_{r=1}^R p_r \int [F_r(x, x')-F(x, x')]^2d\nu(x)d\nu(x'),
\end{equation}
where $p_r = P(Y = y_r)> 0$ for $r = 1, \ldots, R$.

We now provide a consistent estimator for $MDD(X|Y)$. Suppose that $\{(X_i,Y_i) : i = 1, \ldots, n\}$ with the sample size $n$ are $i.i.d.$ samples randomly drawn from the population distribution of $(X,Y)$. $n_r = \sum_{i=1}^nI(Y_i=y_r)$ denotes the sample size of the rth class and $\hat{p}_r = n_r/n$ denotes the sample proportion of the rth class, where $I(\cdot)$ represents the indicator function. Let $\hat{F_r}(x,x')= \frac{1}{n_r}\sum_{k=1}^nI(X_k\in \bar{B}(x,x'),Y_k=y_r),$ and $\hat{F}(x,x') = \frac{1}{n}\sum_{k=1}^nI(X_k \in \bar{B}(x,x'))$. The estimator of $MDD(X|Y)$ can be obtained by the following statistics

\begin{equation}
\begin{aligned}
\label{stats}
\widehat{MDD}(X|Y)& =\frac{1}{n^2}\sum_{r=1}^R\sum_{i=1}^n\sum_{j=1}^n
\hat{p}_r[\hat{F}_r(X_i,X_j)-\hat{F}(X_i,X_j)]^2.
\end{aligned}
\end{equation}

\subsection{Theoretical properties}
In this subsection, we discuss some sufficient conditions for the metric distributional discrepancy and its theoretical properties in Polish space, which is a separable completely metric space. First, to obtain the property of $X$ and $Y$ are independent if and only if $ MDD (X\mid Y ) = 0$, we introduce an important concept named directionally $(\epsilon, \eta, L)$-limited \cite{federer2014}.

\begin{definition}\label{def}
A metric $d$ is called directionally $(\epsilon, \eta, L)$-limited at the subset $\mathcal{A}$ of $\mathcal{M}$, if $\mathcal{A} \subseteq \mathcal{M}$, $\epsilon > 0, 0 < \eta \leq 1/3$, $L$ is a positive integer and the following condition holds: if for each $a \in A$, $D \subseteq A \cap \bar{B}(a, \epsilon)$ such that $d(x, c) \geq \eta d(a, c)$ whenever $ b, c \in D (b \neq c), x \in M $ with 
$$d(a, x) = d(a, c), d(x, b) = d(a, b) - d(a, x),$$
then the cardinality of $D$ is no larger than $L$.
\end{definition}

 There are many metric spaces satisfying directionally $(\epsilon, \eta, L)$-limited, such as a finite dimensional Banach space, a Riemannian manifold of class $\geq 2$. However, not all metric spaces satisfy Definition \ref{def}, such as an infinite orthonormal set in a separable Hilbert space $H$, which is verified in \cite{wang2021nonparametric}.

%---------------------------------%
\begin{th1}\label{the1}
Given a probability measure $\nu$ with its support $supp\{\nu\}$ on $(\mathcal{M}, d)$. Let $X$ be a random element with probability measure $\nu$ on $\mathcal{M}$ and $Y$ be a categorical random variable with $R$ classes $\{y_1,y_2,\dots,y_R\}$.
Then $X$ and $Y$ are independent if and only if $ MDD (X\mid Y ) = 0$ if the metric $d$ is directionally $(\epsilon, \eta, L)$-limited at $supp\{\nu\}$.
\end{th1} 

Theorem \ref{the1} introduces the necessary and sufficient conditions for $MDD(X|Y) = 0$ when the metric is directionally $(\epsilon, \eta, L)$-limited at the support set of the probability measure. Next, we extend this theorem and introduce Corollary \ref{col1}, which presents reasonable conditions on the measure or on the metric.

\begin{col}\label{col1}
(a) \textbf{Measure Condition}: For $ \forall~\epsilon > 0$, there exists $\mathcal{U} \subset \mathcal{M}$ such that $\nu(\mathcal{U}) \geq 1 - \epsilon$ and the
metric $d$ is directionally $(\epsilon, \eta, L)$-limited at $\mathcal{U}$. Then $X$ and $Y$ are independent if and only if $ MDD (X\mid Y ) = 0$.\\
(b) \textbf{Metric Condition}: Given a point $x\in\mathcal{M}$, we define a projection on $\mathcal{M}_k$, $\pi_k(\cdot): \mathcal{M} \to \mathcal{M}_k$ and $\pi_k(x)=x_k$. For a set $A\subset\mathcal{M}$, define $\pi_k(A)=\bigcup\limits_{x\in A}\{\pi_k(x)\}$. There exist $\{(\mathcal{M}_{l}
, d)\}^{\infty}_{l= 1} $ which are the increasing subsets of $\mathcal{M}$, where each $\mathcal{M}_l$
is a Polish space satisfying the directionally-limited condition and their closure $\overline{\bigcup^{\infty}_{l=1}\mathcal{M}_l} =\mathcal{M}$. For every $x \in \mathcal{M}$, $\pi_k(x)$ is unique such that $d(x, \pi_l(x)) = \inf_{z\in\mathcal{M}_l} d(x, z)$ and
$\pi_l|_{\mathcal{M}_{l'}} \circ \pi_{l'} = \pi_l$ if $l' > l$. Then $X$ and $Y$ are independent if and only if $ MDD (X\mid Y ) = 0$.
\end{col}

\begin{th1}\label{the1.2}
$\widehat{MDD}(X|Y)$ almost surely converges to $MDD(X|Y)$.
\end{th1}
Theorem \ref{the1.2} demonstrates the consistency of proposed estimator for $MDD (X\mid Y )$. Hence, $\widehat{MDD}(X|Y)$ is consistent to the metric distributional discrepancy index. Due to the property, we propose an independence test between $X$ and $Y$ based on the MDD index. We consider the following hypothesis testing:
$$
\begin{aligned}
&H_0 : X \text{ and } Y \text{ are statistically independent.} \\
vs\  &H_1: X \text{ and } Y \text{ are not statistically independent.}
\end{aligned}
$$

%------------------------------------------%
\begin{th1}\label{the1.3}
Under the null hypothesis $H_0$,
  \begin{equation*}
    n\widehat{MDD}(X|Y)\stackrel{d}{\rightarrow}\sum_{j=1}^\infty\lambda_j \mathcal{X}_{j}^2(1),
  \end{equation*}
where $\mathcal{X}_j^2(1)$'s, $j = 1, 2,\ldots$, are independently and identically chi-square distribution with $1$ degrees of freedom, and $\stackrel{d}{\rightarrow}$ denotes the convergence in distribution.
\end{th1}
\begin{remark}
In practical application, we can estimate the null distribution of MDD by the permutation when the sample sizes are small, and by the Gram matrix spectrum in \cite{gretton2009fast} when the sample sizes are large.
\end{remark}

\begin{th1}\label{the1.4}
    Under the alternative hypothesis $H_1$,
  \begin{equation*}
    n\widehat{MDD}(X|Y)\stackrel{a.s.}{\longrightarrow}\infty,
  \end{equation*}
where $\stackrel{a.s.}{\longrightarrow}$ denotes the almost sure convergence.
\end{th1}

%-------------------------------------%
\begin{th1}\label{the1.5}
    Under the alternative hypothesis $H_1$,
  \begin{equation*}
    \sqrt{n}[\widehat{MDD}(X|Y)-MDD(X|Y)]\stackrel{d}{\rightarrow}N(0,\sigma^2),
  \end{equation*}
where $\sigma^2$ is given in the appendix.
\end{th1}

\section{Numerical Studies}
\label{secexp}
\subsection{Monte Carlo Simulation}
In this section, we perform four simulations to evaluate the finite sample performance of MDD test by comparing with other existing tests: The distance covariance test (DC) \cite{articleSze} and the HHG test based on pairwise distances \cite{articleheller}. We consider the directional data on the unit sphere $R^p$, which is denoted by $S^{p-1}=\{x\in R^p:||x||_2=1\}$ for $x$ and n-dimensional data independently from Normal distribution. For all of methods, we use the permutation test to obtain p-value for a fair comparison and run simulations to compute the empirical Type-I error rate in the significance level of $\alpha =0.05$. All numerical experiments are implemented by R language. The DC test and the HHG test are conducted respectively by R package $energy$ \cite{energy} and R package $HHG$ \cite{hhg}.
\begin{table*}
\centering
\caption{Empirical Type-I error rates at the significance level 0.05 in Simulation 1}
\label{re1}
\begin{tabular}{@{}cc|c|c|c@{}}
\toprule
 &  & $X=(1,\theta,\phi)$ & $X \sim M_3(\mu,k)$ & $X = (X_1, X_2, X_3)$ \\
\midrule
 R& n &\ MDD \ \  DC \ \  HHG &\ MDD \ \  DC \ \  HHG &\ MDD   \ \  DC   \ \  HHG  \\
\midrule
             & 40           & 0.080 0.085 0.070     & 0.040 0.050 0.030  & 0.060    0.060   0.030     \\
             & 60           & 0.025 0.035 0.030     & 0.075 0.085 0.065  & 0.070    0.045   0.045      \\
2            & 80           & 0.030 0.045 0.065     & 0.035 0.035 0.030  & 0.035   0.035   0.045   \\
             & 120          & 0.040 0.050 0.050     & 0.055 0.050 0.055  & 0.025   0.025   0.025     \\
             & 160          & 0.035 0.035 0.055     & 0.045 0.040 0.050  & 0.050   0.075   0.065      \\
\midrule
             & 40           & 0.015 0.025 0.045     & 0.050 0.045 0.050  & 0.050 0.060 0.055     \\
             & 60           & 0.045 0.025 0.025     & 0.050 0.030 0.035  & 0.050 0.050 0.070       \\
5            & 80           & 0.035 0.030 0.050     & 0.060 0.070 0.060  & 0.060 0.066 0.065    \\
             & 120          & 0.040 0.040 0.035     & 0.050 0.050 0.050  & 0.045 0.030 0.070      \\
             & 160          & 0.030 0.065 0.045     & 0.050 0.060 0.040  & 0.050 0.035 0.035      \\
\bottomrule
\end{tabular}
\end{table*}

\textbf{Simulation 1} In this simulation, we test independence between a high-dimensional variable and a categorical one. We randomly generate three different types of data $X$ listed in three columns in Table \ref{re1}. For the first type in column 1, we set $p=3$ and consider coordinate of $S^2$, denoted as $(r,\theta,\phi)$ where $r = 1$ as radial distance, $\theta$ and $\phi$ were simulated from the Uniform distribution $U(-\pi,\pi)$. For the second type in column $2$, we generate three-dimensional variable from von Mises-Fisher distribution $M_3(\mu, k)$ where $\mu = (0,0,0)$ and $k = 1$. For the third type in column 3, each dimension of $X$ is independently formed from $N(0,1)$ where $X = (X_1, X_2, X_3), X_i \sim N(0,1)$. We generate the categorical random variable $Y$ from $R$ classes ${1,2, \dots, R}$ with the unbalanced proportion $p_r = P(Y = r) = 2[1 + (r-1)/(R-1)]/3R,\ r = 1,2, \dots, R$. For instance, when $Y$ is binary, $p_1 = 1/3$ and $p_2 = 2/3$ and when $R = 5, Y={1,2,3,4,5}$. Simulation times is set to 200. The sample size $n$ are chosen to be 40, 60, 80, 120, 160. The results summarized in Table \ref{re1} show that all three tests perform well in independence testing since empirical Type-I error rates are close to the nominal significance level even in the condition of small sample size.

\begin{table*}
\centering
\caption{Empirical powers at the significance level 0.05 in Simulation 2}
\label{re2}
\begin{tabular}{@{}cc|c|c|c@{}}
\toprule
 &  & $X=(1,\theta,\phi)$ & $X \sim M_3(\mu,k)$ & $X = (X_1, X_2, X_3)$ \\
\midrule
R & n&\ MDD \ \  DC \ \  HHG &\ MDD \ \  DC \ \  HHG &\ MDD \ \  DC \ \  HHG\\
 \midrule
             & 40           & 0.385 0.240 0.380    & 0.595 0.575 0.590      & 0.535 0.685 0.375  \\
             & 60           & 0.530 0.385 0.475    & 0.765 0.745 0.650      & 0.730 0.810 0.590   \\
2           & 80           & 0.715 0.535 0.735    & 0.890 0.880 0.775      & 0.865 0.930 0.755   \\
             & 120          & 0.875 0.740 0.915    & 0.965 0.955 0.905     & 0.970 0.990 0.925    \\
             & 160          & 0.965 0.845 0.995    & 1.000 1.000 0.985     & 0.995 0.995 0.980    \\

\midrule
             & 40           & 0.925 0.240 0.940    & 0.410 0.230 0.185     & 0.840 0.825 0.360     \\
             & 60           & 0.995 0.350 0.995    & 0.720 0.460 0.330     & 0.965 0.995 0.625       \\
5            & 80           & 1.000 0.450 1.000    & 0.860 0.615 0.540     & 0.995 0.990 0.850    \\
             & 120          & 1.000 0.595 1.000    & 0.990 0.880 0.820     & 1.000 0.995 0.990     \\
             & 160          & 1.000 0.595 1.000    & 0.995 0.975 0.955     & 1.000 1.000 1.000      \\
\bottomrule
\end{tabular}
\end{table*}

\textbf{Simulation 2 }
    In this simulation, we test the dependence between a high-dimensional variable and a categorical random variable when $R = 2$ or $5$ with the proportion proposed in simulation 1. In column 1, we generate $X$ and $Y$ representing radial data as follows:  
$$ 
\begin{aligned}
\text{(1)  \ } &Y= \{1,2\}, (a)\ X=(1,\theta,\phi_1 + \epsilon), \\
&\theta \sim U(-\pi,\pi),\  \phi_1\sim U(-\pi,\pi), \epsilon = 0, \\
&(b)\ X=(1,\theta,\phi_2+\epsilon),\\
 &\theta \sim U(-\pi,\pi), \phi_2 \sim U(1/5\pi,4/5\pi), \epsilon \sim t(0,1).
\end{aligned}
$$
$$ 
\begin{aligned}
\text{(2)  \   } &Y=\{1,2,3,4,5\}, X=(1,\theta,\phi_r + \epsilon),\\
 & \phi_r\sim U([-1+2(r-1)/5]\pi,(-1+2r/5) \pi), \\
&r = 1,2,3,4,5.\ \ \theta \sim U(-\pi,\pi),\  \epsilon \sim t(0,1).
\end{aligned}
$$
    For column 2, we consider von Mises-Fisher distribution. Then we set $n = 3$ and $k = 1$, and the simulated data sets are generated as follows: 
$$ 
\begin{aligned}
& (1) R = 2, \mu = (\mu_1,\mu_2) = (1,2), \\
& (2) R = 5,\mu =\{\mu_1,\mu_2,\mu_3,\mu_4,\mu_5\} = \{4,3,1,5,2\}.
\end{aligned}
$$
For column 3, $X$ in each dimension was separately generated from normal distribution $N(\mu, 1)$ to represent data in the Euclidean space. There are two choices of $R$:
$$ 
\begin{aligned}
 &(1) R = 2, \mu = (\mu_1,\mu_2) = (0, 0.6),\\
 &(2) R=5,\mu =\{\mu_1,\mu_2,\mu_3,\mu_4,\mu_5\} = \{4,3,1,5,2\}/3.
 \end{aligned}
$$
Table \ref{re2} based on 200 simulations shows that the MDD test performs well in most settings with Type-I error rate approximating 1 especially when the sample size exceeds 80. When data contains extreme value, the Type-I error rate of DC deteriorate quickly while the MDD test performs more stable. Moreover, the MDD test performs better than both DC test and HHG test in spherical space, especially when the number of class $R$ increases.

\textbf{Simulation 3 }
    In this simulation, we set $X$ with different dimensions, with the range of $\{3,6,8,10,12\}$, to test independence and dependence between a high-dimensional random variable and a categorical random variable. Respectively, I represents independence test, and II represents dependence one. We use empirical type-I error rates for I, empirical powers for II. We let sample size $n = 60$ and classes $R = 2$. Three types of data are shown in Table \ref{re3} for three columns. \\
    In column 1: 
    $$
\begin{aligned}
    X_{dim}&=(1,\theta,\phi_1+\epsilon, \dots, \phi_d + \epsilon), \\
    &d = 1,4,6,8,10, \theta \sim U(-\pi,\pi) \\
     \end{aligned}
$$
 For $\phi_i$ of two classes, 
 $$(a) \phi_i\sim U(-\pi,\pi), \epsilon = 0, $$
 $$ (b) \phi_i\sim U(-1/5\pi,4/5\pi), \epsilon \sim t(0,1),$$
 In column 2, $$X \sim M_d(\mu, k), d= 3,6,8,10,12, \mu \in \{0,2\} ,$$
    where $k=1$.\\
    In column 3, we set $$X_i=(x_{i1},x_{i2}, \dots, x_{id}), d = 3,6,8,10,12 \text{ and } x_{id}\sim N(\mu,1),$$ where $\mu\in\{0,\frac{3}{5}\}$ in dependence test.  Table \ref{re3} based on 300 simulations at $\alpha = 0.05$ shows that the DC test performs well in normal distribution but is conservative for testing the dependence between a radial spherical vector and a categorical variable. The HHG test works well for extreme value but is conservative when it comes to von Mises-Fisher distribution and normal distribution. It also can be observed that MDD test performs well in circumstance of high dimension.
\begin{table*}
\centering
\caption{Empirical Type-I error rates / Empirical powers at the significance level 0.05 in Simulation 3}
\label{re3}
\begin{tabular}{@{}cc|c|c|c@{}}
\toprule
 &  & $X_{dim}=(1,\theta,\phi_d)$ & $X\sim M_{dim}(\mu,k)$ & $X=(x_{1},x_{2}, \dots, x_{dim})$ \\
\midrule
 & dim&\ MDD \ \  DC \ \  HHG &\ MDD \ \  DC \ \  HHG &\ MDD \ \  DC \ \  HHG\\
 \midrule
             & 3           & 0.047 0.067 0.050      & 0.047 0.047 0.053      & 0.047 0.060 0.033     \\
             & 6           & 0.053 0.050 0.050      & 0.050 0.053 0.037      & 0.047 0.050 0.030       \\
I            & 8           & 0.037 0.040 0.050      & 0.037 0.047 0.050      & 0.040 0.037 0.030    \\
             & 10          & 0.050 0.053 0.040      & 0.050 0.056 0.033      & 0.037 0.050 0.060     \\
             & 12          & 0.020 0.070 0.050      & 0.053 0.047 0.047      & 0.040 0.067 0.043      \\
\midrule
             & 3           & 0.550 0.417 0.530      & 0.967 0.950 0.923         & 0.740 0.830 0.577    \\
             & 6           & 1.000 0.583 0.997      & 0.983 0.957 0.927         & 0.947 0.963 0.773      \\
II           & 8           & 1.000 0.573 1.000      & 0.967 0.950 0.920         & 0.977 0.990 0.907  \\
             & 10          & 1.000 0.597 1.000      & 0.977 0.950 0.887         & 0.990 0.993 0.913     \\
             & 12          & 1.000 0.590 1.000      & 0.957 0.940 0.843         & 0.997 1.000 0.963     \\
\bottomrule
\end{tabular}
\end{table*}

\begin{table*}
\centering
\caption{Empirical powers at the significance level 0.05 in Simulation 4}
\label{re4}
\begin{tabular}{@{}cc|c|c|c@{}}
\toprule
 &  & landmark =20   & landmark =50  & landmark =70 \\
\midrule
 & $corr$   &\ MDD \ \  DC \ \  HHG &\ MDD \ \  DC \ \  HHG  &\ MDD \ \  DC \ \  HHG \\
 \midrule
             & 0           & 0.050 0.050 0.063      & 0.047 0.047 0.050     & 0.040 0.050 0.050   \\
             & 0.05      & 0.090 0.080 0.057      & 0.097 0.073 0.087     & 0.093 0.067 0.467   \\
             & 0.10      & 0.393 0.330 0.300      & 0.347 0.310 0.237     & 0.350 0.303 0.183   \\
             & 0.15      & 0.997 0.957 0.997      & 0.993 0.943 0.970     & 0.993 0.917 0.967   \\
             & 0.20      & 1.000 1.000 1.000      & 1.000 1.000 1.000      & 1.000 1.000 1.000  \\
\bottomrule
\end{tabular}
\end{table*}

\textbf{Simulation 4 }
    In this simulation, we use the $(\cos(\theta + d/2)+\epsilon/10, \cos(\theta - d/2)+\epsilon/10) $ parametrization of an ellipse where $\theta \in (0,2\pi)$ to run our experiment, let $X$ be an ellipse shape and $Y$ be a categorical variable. The $\cos(d)$ is the parameter of correlation, which means when $\cos(d) = 0 $, the shape is a unit circle and when $\cos(d) = 1$, the shape is a straight line. We set the noise $\epsilon \sim t(2)$ and $R=2$ where $y_1=1$ represents that shape $X$ is a circle with $\cos(d) = 0$ and $y_2 = 2$ represent that shape $X$ is an ellipse with correlation $ corr = \cos(d)$. It is intuitive that, when $corr = 0$, the MDD statistics should be zero. In our experiment, we set corr = 0,0.05,0.1,0.15,0.2 and let the number of landmark comes from $\{20,50,70\}$. Sample size is set to 60 and R package $shapes$ \cite{shapes} is used to calculate the distance between shapes. Table \ref{re4} summarizes the empirical Type-I error based on 300 simulations at $\alpha = 0.05$. It shows that the DC test and HHG test are conservative for testing the dependence between an ellipse shape and a categorical variable while the MDD test works well in different number of landmarks.

\section{A Real-data Analysis}
\label{realdata}
\subsection{The Hippocampus Data Analysis}
Alzheimer's disease (AD) is a disabling neurological disorder that afflicts about $11\%$ of the population over age 65 in United States. It is an aggressive brain disease that cause dementia -- a continuous decline in thinking, behavioral and social skill that disrupts a person's ability to function independently. As the disease progresses, a person with Alzheimer's disease will develop severe memory impairment and lost ability to carry out the simplest tasks. There is no treatment that cure Alzheimer's diseases or alter the disease process in the brain.

The hippocampus, a complex brain structure, plays an important role in the consolidation of information from short-term memory to long-term memory. Humans have two hippocampi, each side of the brain. It is a vulnerable structure that gets affected in a variety of neurological and psychiatric disorders \cite{hippo2}. In Alzheimer's disease, the hippocampus is one of the first regions of the brain to suffer damage \cite{damage}. The relationship between hippocampus and AD has been well studied for several years, including volume loss of hippocampus \citep{loss, hippo1}, pathology in its atrophy \cite{patho} and genetic covariates \cite{gene}. For instance, shape changes \cite{shapechanges, diagnosis1, diagnosis2} in hippocampus are served as a critical event in the course of AD in recent years.

We consider the radical distances of hippocampal 30000 surface points on the left and right hippocampus surfaces. In geometry, radical distance, denoted $r$, is a coordinate in polar coordinate systems $(r, \theta)$. Basically, it is the scalar Euclidean distance between a point and the origin of the system of coordinates. In our data, the radical distance is the distance between medial core of the hippocampus and the corresponding vertex on the surface. The dataset obtained from the ADNI (The Alzheimer’s Disease Neuroimaging Initiative) contains 373 observations (162 MCI individuals transformed to AD and the 212 MCI individuals who are not converted to AD) where Mild Cognitive Impairment (MCI) is a transitional stage between normal aging and the development of AD \cite{petersen} and 8 covariates in our dataset. Considering the large dimension of original functional data, we firstly use SVD (Singular value decomposition) to extract top 30 important features that can explains 85\% of the total variance.

We first apply the MDD test to detect the significant variables associated with two sides of hippocampus separately at significance level $\alpha = 0.05$. Since 8 hypotheses are simultaneously tested, the Benjamini–Hochberg (BH) correction method is used to control false discover rate at 0.05, which ranks the p-value from the lowest to the highest. The statistics in \eqref{stats} is used to do dependence test between hippocampus functional data and categorical variables. The categorical variables include Gender (1=Male; 2=Female), Conversion (1=converted, 0=not converted), Handedness (1=Right; 2=Left), Retirement (1=Yes; 0=No). Then, we apply DC test and HHG test for the dataset. Note that the p-value obtained in the three methods all used permutation test with 500 times. Table \ref{hippocampi_covariates} summarizes the results that the MDD test, compared to other methods, are able to detect the significance on both side of hippocampus, which agrees with the current studies \cite{hippoage,diagnosis1,diagnosis2} that conversion and age are critical elements to AD disease. Then, we expanded our method to continuous variables, the Age and the ADAS-Cog score, which are both important to diagnosing AD disease \cite{hippoage,kong}. We discretize age and ADAS-Cog score into categorical ones by using the quartile. For instance, the factor level labels of Age are constructed as "$(54,71],(71,75],(75,80],(80,90]$", labelled as $1,2,3,4$. The result of p-values in Table \ref{hippocampi_covariates} agrees with the current research.

Next, we step forward to expand our method to genetic variables to further check the efficiency of MDD test. Some genes in hippocampus are found to be critical to cause AD, such as The apolipoprotein E gene (APOE). The three major human alleles (ApoE2, ApoE3, ApoE4) are the by-product of non-synonymous mutations which lead to changes in functionality and are implicated in AD. Among these alleles, ApoE4, named $\epsilon_4$, is accepted as a factor that affect the event of AD \cite{apoe41,apoe42,apoe43}. In our second experiment, we test the correlation between ApoE2($\epsilon_2$), ApoE3($\epsilon_3$), ApoE4($\epsilon_4$) and hippocampus. The result in the Table \ref{hippocampi_APOE} agrees with the idea that  $\epsilon_4 $ is a significant variable to hippocampus shape. Besides, due to hippocampal asymmetry, the $\epsilon_4$ is more influential to the left side one. The MDD test performs better than DC test and HHG test for both sides of hippocampus.

From the two experiments above, we can conclude that our method can be used in Eu variables, such as age, gender. It can also be useful and even better than other popular methods when it comes to genetic covariates. The correlation between genes and shape data(high dimensional data) is an interesting field that hasn't been well studied so far. There are much work to be done in the future.

Finally, we apply logistic regression to the hippocampus dataset by taking the conversion to AD as the response variable and the gender, age, hippocampus shape as predict variables. The result of regression shows that the age and hippocampus shape are significant and we present the coefficients of hippocampus shape in Figure \ref{fig1} where a blue color indicates positive regions.

\begin{table}[ht]
\centering
\caption{The p-values for correlating hippocampi data and covariates after BH correction.}

\begin{tabular}{@{}c|c|c@{}}
\toprule
    &  left &  right \\
 \midrule
   covariates  & MDD   DC  HHG & MDD   DC  HHG \\
\midrule
  Gender                              & 0.032 0.014  0.106 & 0.248 0.036 0.338 \\
  \textbf{Conversion to AD}  & 0.012 0.006 0.009 & 0.008 0.006 0.009 \\
  Handedness                      &  0.600 0.722 0.457 & 0.144 0.554 0.045 \\
  Retirement                         & 0.373 0.198 0.457 & 0.648 0.267 0.800 \\
  \textbf{Age}                        & 0.012 0.006  0.009 & 0.009 0.006 0.009 \\
    \textbf{ADAS-Cog Score} & 0.012 0.006  0.012 & 0.012 0.006 0.012 \\
 \bottomrule
\end{tabular}\label{hippocampi_covariates}
\end{table}

\begin{table}[ht]
\centering
\caption{The p-values for correlating hippocampi data and APOE covariates after BH correction.}
\begin{tabular}{@{}c|c|c@{}}
\toprule
    &  left &  right \\
 \midrule
   APOE  & MDD   DC  HHG & MDD   DC  HHG \\
\midrule
  $\epsilon_2$                          & 0.567 0.488 0.223 & 0.144 0.696 0.310\\
  $\epsilon_3$                          & 0.357 0.198 0.449 & 0.157 0.460 0.338 \\
  \textbf{$\epsilon_4$}  & \textbf{0.022}  0.206  0.036 & 0.072 0.083 0.354\\
 \bottomrule
\end{tabular}\label{hippocampi_APOE}
\end{table}

\begin{figure}
\begin{center}
\includegraphics[width=8cm]{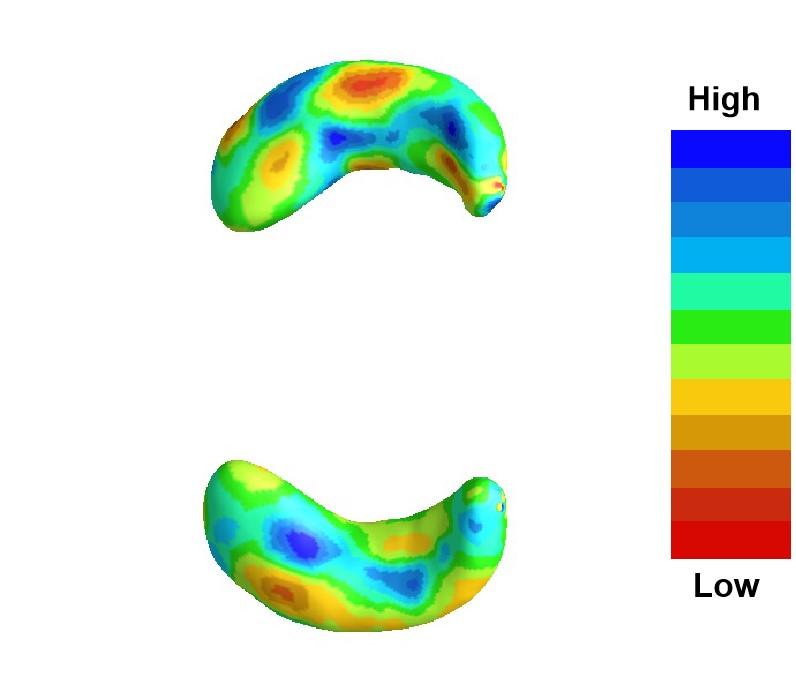}
\end{center}
\caption{The Hippocampus 3D Shape Surface with coefficients}\label{fig1}
\end{figure}

{\subsection{The corpus callosum Data Analysis}
We consider another real data, the corpus callosum (CC), which is the largest white matter structure in the brain. CC has been a structure of high interest in many neuroimaging studies of neuro-developmental pathology. It helps the hemispheres share information, but it also contributes to the spread of seizure impulses from one side of the brain to the other.
Recent research  \cite{raz2010trajectories,witelson1989hand} has investigated the individual differences in CC and their possible implications regarding interhemispheric connectivity for several years.

We consider the CC contour data obtained from the ADNI study to test the dependence between a high-dimensional variable and a random variable. In the ADNI dataset, the segmentation of the T1-weighted MRI and the calculation of the intracranial volume were done by using $FreeSurfer$ package created by \cite{dale1999cortical}, whereas the midsagittal CC area was calculated by using CCseg package. The CC data set includes 409 subjects with 223 healthy controls and 186 AD patients at baseline of the ADNI database. Each subject has a CC planar contour $Y_j$ with 50 landmarks and five covariates. We treat the CC planar contour $Y_j$ as a manifold-valued response in the Kendall planar shape space and all covariates in the Euclidean space. The Riemannian shape distance was calculated by R package $shapes$ \cite{shapes}.

 It is of interest to detect the significant variable associated with CC contour data. We applied the MDD test for dependence between five categorical covariates, gender, handedness, marital status, retirement and diagnosis at the significance level $\alpha = 0.05.$ We also applied DC test and HHG test for the CC data. The result is summarized in Table \ref{realt1}.  It reveals that the shape of CC planar contour are highly dependent on gender, AD diagnosis. It may indicate that gender and AD diagnosis are significant variables, which agree with \cite{witelson1989hand,pan2017conditional}. This result also demonstrated that the MMD test performs better to test the significance of variable gender than HHG test. 
 
 We plot the mean trajectories of healthy controls (HC) and Alzheimer's disease (AD). The similar process is conducted on the Male and Female. Both of the results are shown in Figure \ref{Fig:CCdata}.
It can be observed that there is an obvious difference of the shape between the AD disease and healthy controls. Compared to healthy controls, the spleen of
AD patients seems to be less thinner and the isthmus is more rounded. Moreover, the splenium can be observed that it is thinner in male groups than in female groups. This could be an intuitive evidence to agree with the correlation between gender and the AD disease.

\begin{table}[ht]
\centering
\caption{The p-values for correlating CC contour data and five categorical covariates after BH correction.}\label{realt1}
\begin{tabular}{ c|c|c|p{1cm} }
\toprule
  covariates & MDD & DC &HHG \\ 
\midrule
  Gender & 0.015 &0.018  & 0.222\\ 
  Handedness & 0.482 &0.499& 0.461 \\ 
  Marital Status & 0.482 &0.744 & 0.773\\ 
  Retirement & 0.482 &0.482 & 0.461\\ 
 Diagnosis & 0.015&0.018 & 0.045\\ 

\bottomrule
\end{tabular}
\end{table}

\begin{figure}
\begin{center}
\includegraphics[width=8cm]{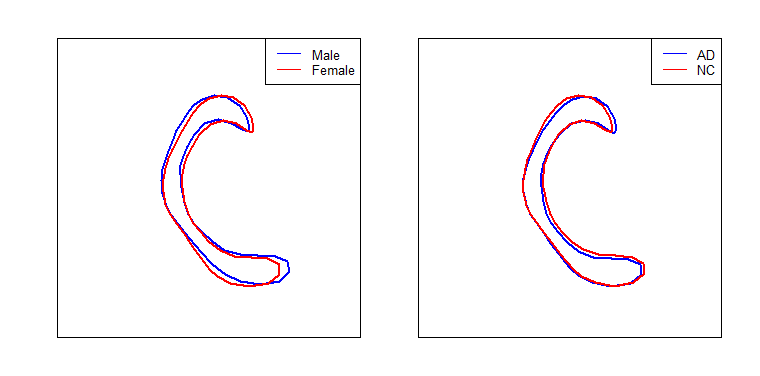}
\end{center}
\caption{The corpus callosum Data Surface}\label{Fig:CCdata}
\end{figure}
}
%%% Please be aware that for original research articles we only permit a combined number of 15 figures and tables, one figure with multiple subfigures will count as only one figure.
%%% Use this if adding the figures directly in the mansucript, if so, please remember to also upload the files when submitting your article
%%% There is no need for adding the file termination, as long as you indicate where the file is saved. In the examples below the files (logo1.eps and logos.eps) are in the Frontiers LaTeX folder
%%% If using *.tif files convert them to .jpg or .png
%%%  NB logo1.eps is required in the path in order to correctly compile front page header %%%

\section{Conclusion}
\label{conclu}
In this paper, we propose the MDD statistics of correlation analysis for Non-Euclidean data in metric spaces and give some conditions for constructing the statistics. Then, we proved the mathematical preliminaries needed in our analysis. The proposed method is robust to outliers or heavy tails of the high dimensional variables. Depending on the results of simulations and real data analysis in hippocampus dataset from the ADNI, we demonstrate its usefulness for detecting correlations between the high dimensional variable and different types of variables (including genetic variables). We also demonstrate its usefulness in another manifold-valued data, CC contour data. We plan to explore our method to variable selection methods and other regression models.

\section*{Acknowledgements}
 W.-L. Pan was supported by the Science and Technology Development Fund, Macau SAR (Project code: 0002/2019/APD), National Natural Science Foundation of China (12071494), Hundred Talents Program of the Chinese Academy of Sciences and the Science and Technology Program of Guangzhou, China (202002030129,202102080481). P. Dang was supported by the Science and Technology Development Fund, Macau SAR (Project code: 0002/2019/APD), the Science and Technology Development Fund, Macau SAR(File no. 0006/2019/A1). W.-X. Mai was supported by NSFC Grant No. 11901594.

Data collection and sharing for this project was funded by the Alzheimer's Disease Neuroimaging Initiative (ADNI) (National Institutes of Health Grant U01 AG024904) and DOD ADNI (Department of Defense award number W81XWH-12-2-0012). ADNI is funded by the National Institute on Aging, the National Institute of Biomedical Imaging and Bioengineering, and through generous contributions from the following: Alzheimer's Association; Alzheimer's Drug Discovery Foundation; Araclon Biotech; BioClinica, Inc.; Biogen Idec Inc.; Bristol-Myers Squibb Company; Eisai Inc.; Elan Pharmaceuticals, Inc.; Eli Lilly and Company; EuroImmun; F. Hoffmann-La Roche Ltd. and its affiliated company Genentech, Inc.; Fujirebio; GE Healthcare; IXICO Ltd.; Janssen Alzheimer Immunotherapy Research \& Development, LLC.; Johnson \& Johnson Pharmaceutical Research \& Development LLC.; Medpace, Inc.; Merck \& Co., Inc.; Meso Scale Diagnostics, LLC.; NeuroRx Research; Neurotrack Technologies; Novartis Pharmaceuticals Corporation; Pfizer Inc.; Piramal Imaging; Servier; Synarc Inc.; and Takeda Pharmaceutical Company. The Canadian Institutes of Health Research is providing funds to support ADNI clinical sites in Canada. Private sector contributions are facilitated by the Foundation for the National Institutes of Health (www.fnih.org). The grantee organization is the Northern California Institute for Research and Education, and the study is coordinated by the Alzheimer's Disease Cooperative Study at the University of California, San Diego. ADNI data are disseminated by the Laboratory for Neuro Imaging at the University of Southern California. We are grateful to Prof. Hongtu Zhu for generously sharing preprocessed ADNI dataset to us.

\setcounter{equation}{0}
\setcounter{subsection}{0}
\renewcommand{\theequation}{A.\arabic{equation}}
\renewcommand{\thesubsection}{A.\arabic{subsection}}
\renewcommand{\thedefinition}{A.\arabic{definition}}
\renewcommand{\thelemma}{A.\arabic{lemma}}
\section*{Appendix: Technical details}
%\label{appendix}

\begin{proof}[\textbf{Proof of Theorem \ref{the1}}]
% We give a covering theorem\cite{federer2014} first in order to prove Theorem \ref{the1}.
% \begin{lemma}\label{cover}
% Let $\nu$ be a Borel measure defined on $(\mathcal{M}, d).$ Suppose that $\mathcal{F}$ is a collection of non-degenerate closed balls, the set of centers of which is denoted by $A$, such that $\nu(A)<\infty$, and $A \subset \mathcal{M}$ is a directionally- $(\epsilon, \eta, L)$ limited subset of $\mathcal{M}$. For every $a \in A$ and every $\delta>0, \mathcal{F}$ contains a ball $\bar{B}(a, r)$ with $r<\delta$. Then, for every nonempty open set $U \subseteq \mathcal{M}$, one can find at most countable collection of disjoint balls $\bar{B}_{j} \in \mathcal{F}$ such that
% \end{lemma}
% $$
% \bigcup_{j=1}^{\infty} \bar{B}_{j} \subseteq U \text { and } \nu\left((A \cap U) \backslash \bigcup_{j=1}^{\infty} \bar{B}_{j}\right)=0
% $$
%Since we know $MDD(X \mid Y)=\sum_{r=1}^{R} p_{r} \int\left[F_{r}\left(x, x^{\prime}\right)-F\left(x, x^{\prime}\right)\right]^{2} d \nu(x) d \nu\left(x^{\prime}\right)$, where $F(x, x^{\prime}) = P_{X}\left\{X\in \bar{B}\left(x, x^{\prime}\right) \right\}$, $F_{r}(x, x^{\prime}) = P_{X}\{X \in \bar{B}\left(x, x^{\prime}\right) \mid Y= $

%$y_{r}\}$, $(\mathcal{M}, d)$ is a metric space, $X$ is a high dimensional variable in $\mathcal{M}$ and $\nu$ is a Borel probability measure of $X$.

It is obvious that if $X$ and $Y$ are independent, $F_{r}(x, x^{\prime}) = F(x, x^{\prime})$ for $ \forall x, x^{\prime} \in \mathcal{M}$, then $MDD(X \mid Y) = 0$. Next, we need to prove that if $MDD(X \mid Y) = 0$, then $X$ and $Y$ are independent.

According to the definition of $MDD$, $MDD(X\mid Y)=\sum_{r=1}^{R} p_{r} \int [F_r(u, v)-F(u, v)]^{2} d \nu(u) d \nu(v)$. It is obvious that $MDD(X \mid Y) \geq 0$, so if $MDD(X \mid Y)=0$, we have $F_r(u, v) = F(u, v)$, $a.s.$ $\nu \otimes \nu$.

Given $Y=y_{r}$, define $\phi_{r}$ is a Borel probability measure of $X \mid Y=y_{r}$, and we have
$$
\phi_{r}[\bar{B}(u, d(u, v))] := F_{r}(u, v) =P(X \in \bar{B}(u, v) \mid Y=y_{r}).
$$

Because $(\mathcal{M}, d)$ is a Polish space that $d$ is directionally $(\epsilon, \eta, L)$-limited and we have $F_r(u, v) = F(u, v)$, $a.s.$ $\nu \otimes \nu$. Next, we can apply Theorem 1 in \cite{wang2021nonparametric} to get the conclusion that $\nu = \phi_{r}$ for $r = 1, \ldots, R$.

Therefore, we have $F_{r}(x, x^{\prime}) = F(x, x^{\prime})$ for $ \forall x, x^{\prime} \in \mathbb{R}$. That is, for every $x$, $x^{\prime}$ and every $r$, we have $$P(X \in \bar{B}(x, x^{\prime}) \mid Y = y_{r}) = P(X \in \bar{B}(x, x^{\prime})),$$ $i.e.$ $X$ and $Y$ are independent.
\end{proof}

\begin{proof}[\textbf{Proof of Corollary \ref{col1}}]
(a) We have $\nu(\mathcal{U}) \geq 1 - \epsilon$ and the
metric $d$ is directionally $(\epsilon, \eta, L)$-limited at $\mathcal{U}$,  we can obtain the result of independence according to Theorem \ref{the1} and Corollary 1(a) in \cite{wang2021nonparametric}.

(b) Similarly, we know that each $\mathcal{M}_l$
is a Polish space satisfying the directionally-limited condition, we also can obtain the result of independence according to  Corollary 1(b) in \cite{wang2021nonparametric}.
\end{proof}
  
\begin{proof}[\textbf{Proof of Theorem \ref{the1.2}}]

Consider
\begin{equation*}
\begin{aligned}
&\widehat{MDD}(X|Y) = \frac{1}{n^2}\sum_{r=1}^R\sum_{i=1}^n\sum_{j=1}^n \hat{p}_r[\hat{F}_r(X_i,X_j)-\hat{F}(X_i,X_j)]^2 \\
=&\frac{1}{n^2}\sum_{r=1}^R\sum_{i=1}^n\sum_{j=1}^n \hat{p}_r[\frac{1}{n}\sum_{k=1}^n \frac{I(X_k\in\bar{B}(X_i,X_j),Y_k=y_r)}{\hat{p}_r}\\
&\quad\quad\quad\quad\quad\quad\quad
-\frac{1}{n}\sum_{k=1}^n I(X_k\in\bar{B}(X_i,X_j)]^2\\
=&\sum_{r=1}^R\Big[\frac{1}{n^4\hat{p}_r}\sum_{i,j,k,k^\prime=1}^n I(X_k\in\bar{B}(X_i,X_j),Y_k=y_r)\\
&\quad\quad\quad\quad\quad\quad\quad\quad\quad
I(X_{k^\prime}\in\bar{B}(X_i,X_j),Y_{k^\prime}=y_r) \\
&+\frac{\hat{p}_r}{n^4}\sum_{i,j,k,k^\prime=1}^n I(X_k\in\bar{B}(X_i,X_j))I(X_{k^\prime}\in\bar{B}(X_i,X_j)) \\
&-\frac{1}{n^4}\sum_{i,j,k,k^\prime=1}[I(X_k\in\bar{B}(X_i,X_j),Y_k=y_r)I(X_{k^\prime}\in\bar{B}(X_i,X_j))  \\ 
&-\frac{1}{n^4}\sum_{i,j,k,k^\prime=1}I(X_{k^\prime}\in\bar{B}(X_i,X_j),Y_{k^\prime}=y_r)I(X_k\in\bar{B}(X_i,X_j))\Big]\\
=&:Q_1+Q_2+Q_3.
\end{aligned}
\end{equation*}

For $Q_1$, $\frac{1}{n^4}\sum_{i,j,k,k^\prime=1}^n I(X_k\in\bar{B}(X_i,X_j),Y_k=y_r)I(X_{k^\prime}\in\bar{B}(X_i,X_j),Y_{k^\prime}=y_r)$ is a V-statistic of order 4.
% we denote $Z_i=(X_i, Y_i)$, $Z_j=(X_j, Y_j)$, $Z_k=(X_k, Y_k)$, $Z_{k^\prime}=(X_{k^\prime}, Y_{k^\prime})$, then we find 
% with kernel $I(X_k\in\bar{B}(X_i,X_j),Y_k=y_r)I(X_{k^\prime}\in\bar{B}(X_i,X_j),Y_{k^\prime}=y_r)$.
We can verify that
\begin{equation*}
\begin{aligned}
&E[I(X_k\in\bar{B}(X_i,X_j),Y_k=y_r)I(X_{k^\prime}\in\bar{B}(X_i,X_j),Y_{k^\prime}=y_r)]\\
=&E[E[I(X_k\in\bar{B}(X_i,X_j),Y_k=y_r)|X_i,X_j]\\
&\times E[I(X_{k^\prime}\in\bar{B}(X_i,X_j),Y_{k^\prime}=y_r)|X_i,X_j]]\\
=&E[E^2[I(X_k\in\bar{B}(X_i,X_j),Y_k=y_r)|X_i,X_j]]\\
=&p_r^2E[\frac{1}{p_r^2}E^2[I(X_k\in\bar{B}(X_i,X_j),Y_k=y_r)|X_i,X_j]]\\
=&p_r^2E[F_r^2(X_i,X_j)].
\end{aligned}   
\end{equation*}

Since $E|I(X_k\in\bar{B}(X_i,X_j),Y_k=y_r)I(X_{k^\prime}\in\bar{B}(X_i,X_j),Y_{k^\prime}=y_r)|\leq1<\infty$, according to Theorem 3 of \cite{lee2019u}, $\frac{1}{n^4}\sum_{i,j,k,k^\prime=1}^n I(X_k\in\bar{B}(X_i,X_j),Y_k=y_r)I(X_{k^\prime}\in\bar{B}(X_i,X_j),Y_{k^\prime}=y_r)$ almost surely converges to $p_r^2E[F_r^2(X_i,X_j)]$.
And we have $p_r\to\hat{p}_r$, and we can draw conclusion that $Q_1 \stackrel{a.s.}{\longrightarrow} \sum_{r=1}^R p_rE[ F_r^2(X_i,X_j)]$, as $n\to \infty$.

Similarly, we also have $Q_2\stackrel{a.s.}{\longrightarrow}\sum_{r=1}^R p_rE[F^2(X_i,X_j)]$ and $Q_3\stackrel{a.s.}{\longrightarrow} -2\sum_{r=1}^R p_rE[F_r(X_i,X_j)F(X_i,X_j)]$,  as $n\to \infty$.

Because $MDD(X|Y)=E[\sum_{r=1}^R p_r(F_r(X_i,X_j)-F(X_i,X_j))^2]$,
we have
$$
\widehat{MDD}(X|Y)\stackrel{a.s.}{\longrightarrow} MDD(X|Y),~n\to\infty.
$$
\end{proof}
%-------------------------------------%

Before we prove Theorem \ref{the1.3}, we give an lemma and some notations as follows \cite{lee2019u}. 
\begin{lemma}\label{lm6.4}
Let $V_n$ be a V-statistic of order m, where $V_n = n^{-m}\sum_{i_1=1}^n\cdots\sum_{i_m=1}^n h(X_{i_1}, \ldots, X_{i_m})$ and $h(X_{i_1} ,\ldots, X_{i_m})$ is the kernel of $V_n$. For all $1\leq i_1\leq\cdots\leq i_m\leq m$, $E[h(X_{i_1} ,\ldots, X_{i_m})]^2 < \infty$. We have the following conclusions.

(i) If $\zeta_1 = \Var(h_1(X_1)) > 0$, then
$$\sqrt{n}(V_n - E(h(X_{1} ,\ldots, X_{m})))\stackrel{d}{\rightarrow}N(0, m^2\zeta_1),$$
where $\zeta_k =\Var(h_k(X_1,\ldots,X_K))$ and
\begin{align*}
h_k(X_1,\ldots,X_K)=&E[h(X_1,\ldots,X_m)|X_1=x_1,\ldots,X_k=x_k]\\
=&E[h(x_1,\ldots,x_k,X_{k+1},\ldots,X_m)].
\end{align*}

(ii) If $\zeta_1 = 0$ but $\zeta_2 = \Var(h_2(X_1, X_2)) > 0$, then $n(V_n - E(h(X_{1} ,\ldots, X_{m})))\stackrel{d}{\rightarrow}\frac{m(m-1)}{2}\sum_{j=1}^{\infty}\lambda_j\mathcal{X}_{j}^2(1)$, where $\mathcal{X}_j^2(1)$'s, $j = 1, 2,\ldots$, are independently and identically distributed $\mathcal{X}^2$ random variables with $1$ degree of freedom and $\lambda_j$'s meet the condition $\sum_{j=1}^\infty \lambda_j^2=\zeta_2$.\\
\end{lemma}

% We also give a lemma referring to Theorem 2.16 of Chapter 2 in \cite{de1987effect}. 

% \begin{lemma}\label{Vn}
% We consider V-statistics
% $$
% V_{n}(\hat{\delta})=n^{-2} \sum_{i=1}^{n} \sum_{j=1}^{n} h\left(X_{i}, X_{j} ; \hat{\delta}\right) \text {, }
% $$
% where the estimator $\hat{\delta}$ is also an estimator of $\delta$.

% We consider kernels of the form
% $$
% h\left(x_{1}, x_{2} ; \hat{\delta}\right)=\int_{-\infty}^{\infty} g\left(x_{1}, t ; \hat{\delta}\right) g\left(x_{2}, t ; \hat{\delta}\right) d M(t)
% \text {, }
% $$
% for some function $g$ and $M$ a finite positive measure.

% If $\frac{\partial Eg(X,t;\delta)}{\partial \delta}=0$ hold, then we have
% $$
% n V_{n}(\hat{\delta}) \stackrel{d}{\rightarrow} \sum_{k=1}^{\infty} \lambda_{k} \chi_{k}^{2},
% $$
% where $\chi_{k}^{2}, k=1,2, \ldots$, are independent chi-square random variable with 1 degree of freedom.
% \end{lemma}

\begin{proof}[\textbf{Proof of Theorem \ref{the1.3}}]

Denote $Z_i=(X_i, Y_i)$, $Z_j=(X_j, Y_j)$, $Z_k=(X_k, Y_k)$, $Z_{k^\prime}=(X_{k^\prime}$, $Y_{k^\prime})$. We consider the statistic with the known parameter $p_r$ 
\begin{equation*}
\begin{aligned}
I_n=&\frac{1}{n^2}\sum_{r=1}^R\sum_{i=1}^n\sum_{j=1}^n p_r[\frac{1}{n}\sum_{k=1}^n \frac{I(X_k\in\bar{B}(X_i,X_j),Y_k=y_r)}{p_r}\\
&\quad\quad\quad\quad\quad\quad\quad
-\frac{1}{n}\sum_{k=1}^n I(X_k\in\bar{B}(X_i,X_j)]^2\\
=&\frac{1}{n^4}\sum_{r=1}^R\sum_{i=1}^n\sum_{j=1}^n\sum_{k,k^\prime=1}^n[\frac{1}{p_r}I(X_k\in\bar{B}(X_i,X_j),Y_k=y_r)\\
&\quad\quad\quad\quad\quad\quad\quad\quad\quad\times
I(X_{k^\prime}\in\bar{B}(X_i,X_j),Y_{k^\prime}=y_r) \\
&+p_r I(X_k\in\bar{B}(X_i,X_j))I(X_{k^\prime}\in\bar{B}(X_i,X_j)) \\
&-I(X_k\in\bar{B}(X_i,X_j),Y_k=y_r)I(X_{k^\prime}\in\bar{B}(X_i,X_j))  \\ 
&-I(X_{k^\prime}\in\bar{B}(X_i,X_j),Y_{k^\prime}=y_r)I(X_k\in\bar{B}(X_i,X_j))].
\end{aligned}
\end{equation*}
Let $V_n^{p_r}=\frac{1}{n^4}\sum_{i,j,k,k^\prime=1}^n \Psi^{(r)}(Z_i,Z_j,Z_k,Z_{k^\prime})$ and
\begin{equation*}
\begin{aligned}
&\Psi^{(r)}(Z_i,Z_j,Z_k,Z_{k^\prime}) \\ =&\frac{1}{p_r}I(X_k\in\bar{B}(X_i,X_j),Y_k=y_r)I(X_{k^\prime}\in\bar{B}(X_i,X_j),Y_{k^\prime}=y_r)\\
&+ p_r I(X_k\in\bar{B}(X_i,X_j))I(X_{k^\prime}\in\bar{B}(X_i,X_j))\\
&-I(X_k\in\bar{B}(X_i,X_j),Y_k=y_r)I(X_{k^\prime}\in\bar{B}(X_i,X_j)) \\
&-I(X_{k^\prime}\in\bar{B}(X_i,X_j),Y_{k^\prime}=y_r)I(X_k\in\bar{B}(X_i,X_j)),
\end{aligned}
\end{equation*}
then we have 
\begin{align*}
    I_n =& \sum_{r=1}^R V_n^{(p_r)}\\
    =&\sum_{r=1}^R\frac{1}{n^4}\sum_{i,j,k,k^\prime=1}^n \Psi^{(r)}(Z_i,Z_j,Z_k,Z_{k^\prime}).
\end{align*}
We would like to use V statistic to obatain the asymptotic properties of $\widehat{MDD}(X|Y)$. We symmetrize the kernel $\Psi^{(r)}(Z_i,Z_j,Z_k,Z_{k^\prime})$ and denote
\begin{equation*}
\begin{aligned}
    &\Psi^{(r)}_S(Z_i,Z_j,Z_k,Z_{k^\prime})\\ =&\frac{1}{4!}\sum_{\tau\in\pi(i,j,k,k^\prime)}\Psi^{(r)}(Z_{\tau(1)},Z_{\tau(2)},Z_{\tau(3)},Z_{\tau(4)}),
\end{aligned}
\end{equation*}
where $\pi(i,j,k,k^\prime)$ are the permutations of $\{i,j,k,k^\prime\}$. Now, the kernel $\Psi^{(r)}_S(Z_i,Z_j,Z_k,Z_{k^\prime})$ is symmetric, and $\frac{1}{n^4}\sum_{i,j,k,k^\prime=1}^n\Psi^{(r)}_S(Z_i,Z_j,Z_k,Z_{k^\prime})$ should be a V-statistic. 
By using the denotation of Lemma \ref{lm6.4}, we should consider $E[\Psi^{(r)}_S(z_i,Z_j,Z_k,Z_{k^\prime})]$, that is, the case where only one random variable fixed its value. And, we have to consider $E[\Psi^{(r)}(z_i,Z_j,Z_k,Z_{k^\prime})]$,  $E[\Psi^{(r)}(Z_i,z_j,Z_k,Z_{k^\prime})]$, 
$E[\Psi^{(r)}(Z_i,Z_j,z_k,Z_{k^\prime})]$ and  $E[\Psi^{(r)}(Z_i,Z_j,Z_k,z_{k^\prime})]$.

We consider
\begin{equation*}
\begin{aligned}
&E[\Psi^{(r)}(z_i,Z_j,Z_k,Z_{k^\prime})]\\
=&\frac{1}{p_r}E[I(X_k\in\bar{B}(x_i,X_j),Y_k=y_r)\\
&\quad\quad I(X_{k^\prime}\in\bar{B}(x_i,X_j),Y_{k^\prime}=y_r)]\\
&+ p_r E[I(X_k\in\bar{B}(x_i,X_j))I(X_{k^\prime}\in\bar{B}(x_i,X_j))]\\
&-E[I(X_k\in\bar{B}(x_i,X_j),Y_k=y_r)I(X_{k^\prime}\in\bar{B}(x_i,X_j))] \\
&- E[I(X_{k^\prime}\in\bar{B}(x_i,X_j),Y_{k^\prime}=y_r)I(X_k\in\bar{B}(x_i,X_j))]\\
=& \frac{1}{p_r}P_{j,k,k^\prime}[X_k,X_{k^\prime}\in\bar{B}(x_i,X_j),Y_k=Y_{k^\prime}=y_r)]\\
&+p_rP_{j,k,k^\prime}[X_k,X_{k^\prime}\in\bar{B}(x_i,X_j)]\\
&- P_{j,k,k^\prime}[X_k,X_{k^\prime}\in\bar{B}(x_i,X_j), Y_k=y_r]\\
&-P_{j,k,k^\prime}[X_k,X_{k^\prime}\in\bar{B}(x_i,X_j), Y_{k^\prime}=y_r],
\end{aligned}
\end{equation*}
where $P_{j,k,k^\prime}$ means the probability of $Z_j$, $Z_k$ and $Z_{k^\prime}$.

Under the null hypothesis $H_0$, $X$ and $Y$ are independent. Then we have
\begin{equation*}
\begin{aligned}
&\frac{1}{p_r}P_{j,k,k^\prime}[X_k,X_{k^\prime}\in\bar{B}(x_i,X_j),Y_k=Y_{k^\prime}=y_r]\\
=&\frac{1}{p_r}P_{j,k,k^\prime}[X_k,X_{k^\prime}\in\bar{B}(x_i,X_j),Y_k=y_r]P_{k^\prime}[Y_{k^\prime}=y_r]\\
=&P_{j,k,k^\prime}[X_k,X_{k^\prime}\in\bar{B}(x_i,X_j),Y_k=y_r],
\end{aligned}
\end{equation*}
and 
\begin{equation*}
\begin{aligned}
&p_rP_{j,k,k^\prime}[X_k,X_{k^\prime}\in\bar{B}(x_i,X_j)]\\
=&P_{j,k,k^\prime}[X_k,X_{k^\prime}\in\bar{B}(x_i,X_j),Y_{k^\prime}=y_r],
\end{aligned}
\end{equation*}
Thus, $E[\Psi^{(r)}(z_i,Z_j,Z_k,Z_{k^\prime})]=0$.

Similarly, we have $E[\Psi^{(r)}(Z_i,z_j,Z_k,Z_{k^\prime})]=0$, $E[\Psi^{(r)}(Z_i,Z_j,z_k,Z_{k^\prime})]=0$ and $E[\Psi^{(r)}(Z_i,Z_j,Z_k,z_{k^\prime})]=0$ because of the independence of $X$ and $Y$ under the null hypothesis $H_0$. 

Next, we consider the case when two random elements are fixed. 
\begin{equation*}
\begin{aligned}
&E[\Psi^{(r)}(Z_i,Z_j,z_k,z_{k^\prime})]\\
=&\frac{1}{p_r}E[I(x_k\in\bar{B}(X_i,X_j),y_k=y_r)\\
&\quad\quad\times
I(x_{k^\prime}\in\bar{B}(X_i,X_j),y_{k^\prime}=y_r)]\\
&+ p_r E[I(x_k\in\bar{B}(X_i,X_j))I(x_{k^\prime}\in\bar{B}(X_i,X_j))]\\
&-E[I(x_k\in\bar{B}(X_i,X_j),y_k=y_r)I(x_{k^\prime}\in\bar{B}(X_i,X_j))] \\
&- E[I(x_{k^\prime}\in\bar{B}(X_i,X_j),y_{k^\prime}=y_r)I(x_k\in\bar{B}(X_i,X_j))]\\
=& \frac{1}{p_r}P_{i,j}[x_k,x_{k^\prime}\in\bar{B}(X_i,X_j),y_k=y_{k^\prime}=y_r)]\\
&+p_rP_{i,j}[x_k,x_{k^\prime}\in\bar{B}(X_i,X_j)]\\
&- P_{i,j}[x_k,x_{k^\prime}\in\bar{B}(X_i,X_j), y_k=y_r]\\
&-P_{i,j}[x_k,x_{k^\prime}\in\bar{B}(X_i,X_j), y_{k^\prime}=y_r],
\end{aligned}
\end{equation*}

% According to the previous idea, we need to verify 
% \begin{equation*}\begin{aligned}
% &\frac{1}{p_r}P_{i,j}[x_k,x_{k^\prime}\in\bar{B}(X_i,X_j),y_k=y_{k^\prime}=y_r]\\
% &=P_{i,j}[x_k,x_{k^\prime}\in\bar{B}(X_i,X_j),y_k=y_r],
% \end{aligned}
% \end{equation*}
% however, $y_k, y_{k^\prime}$ are fixed values so that the above equation cannot be derived independently from $X$ and $Y$. 
$E[\Psi^{(r)}(Z_i,Z_j,z_k,z_{k^\prime})]$ is a non-constant function related to $z_k, z_{k^\prime}$. 

In addition, we know 
\begin{equation*}
\begin{aligned}
&E[\Psi^{(r)}(Z_i,Z_j,Z_k,Z_{k^\prime})]\\
=&\frac{1}{p_r}E[I(X_k\in\bar{B}(X_i,X_j),Y_k=y_r)\\
&\quad\quad\times
I(X_{k^\prime}\in\bar{B}(X_i,X_j),Y_{k^\prime}=y_r)]\\
&+ p_r E[I(X_k\in\bar{B}(X_i,X_j))I(X_{k^\prime}\in\bar{B}(X_i,X_j))]\\
&-E[I(X_k\in\bar{B}(X_i,X_j),Y_k=y_r)I(X_{k^\prime}\in\bar{B}(X_i,X_j))] \\
&- E[I(X_{k^\prime}\in\bar{B}(X_i,X_j),Y_{k^\prime}=y_r)I(X_k\in\bar{B}(X_i,X_j))]\\
=& \frac{1}{p_r}P_{i,j,k,k^\prime}[X_k,X_{k^\prime}\in\bar{B}(X_i,X_j),Y_k=Y_{k^\prime}=y_r)]\\
&+p_rP_{i,j,k,k^\prime}[X_k,X_{k^\prime}\in\bar{B}(X_i,X_j)]\\
&- P_{i,j,k,k^\prime}[X_k,X_{k^\prime}\in\bar{B}(X_i,X_j), Y_k=y_r]\\
&-P_{i,j,k,k^\prime}[X_k,X_{k^\prime}\in\bar{B}(X_i,X_j), Y_{k^\prime}=y_r] = 0,
\end{aligned}
\end{equation*}
The last equation is derived from the independence of $X$ and $Y$. By Lemma \ref{lm6.4} $(ii)$, $V_{n}^{(p_r)}$ is a limiting $\chi^2$-type V statistic.

Now, we consider $V_{n}^{(\hat{p}_r)}$. Let $t=(t_1,t_2)$. In showing the conditions of Theorem 2.16 in \cite{de1987effect} hold, we use 
$$
\begin{aligned}
h(z_1,z_2; p_r)=p_r\int g(z_1,t;p_r)g(z_2,t;p_r)dM(t),
\end{aligned}
$$
% $$
% \begin{aligned}
% &h(x_1,x_2; p_r)=p_r[\hat{F}_r(x_1,x_2)-\hat{F}(x_1,x_2)]^2\\
% =&\hat{p}_r[\frac{1}{\hat{p}_r}\frac{1}{n}\sum_{k=1}^n I(x_k\in\bar{B}(x_1,x_2),Y_k=y_r)\\
% &\quad\quad-\frac{1}{n}\sum_{k=1}^n I(x_k\in\bar{B}(x_1,x_2))]^2\\
% =:&\int g(x_1,t;p_r)g(x_2,t;p_r)dM(t),
% \end{aligned}
% $$
where $g(z,t;\gamma)=\sqrt{p_r}(I(z\in\bar{B}(t_1,t_2),y_1=y_r)/\gamma-I(z\in\bar{B}(t_1,t_2)))$ and $ = \nu\otimes\nu$ is the product measure of $\nu$ with respect to $X$.

Thus, 
\begin{align*}
    \mu(t;\gamma)=&Eg(Z,t;\gamma)\\
    =&\sqrt{\gamma}(P(X\in\bar{B}(t_1,t_2))p_r/\gamma-P(X\in\bar{B}(t_1,t_2)))
\end{align*}
and
\begin{align*}
    \mathbf{d}_1\mu(t;p_r)=p_r^{-\frac{1}{2}}P(X\in\bar{B}(t_1,t_2)).
\end{align*}
The condition of Theorem 2.16 in \cite{de1987effect} can be shown to hold in this case using 
\begin{align*}
   h_* &(Z_1,Z_2)\\
    =\int &[g(Z_1,t;p_r)+\mathbf{d}_1\mu(t;p_r)(I(Y_1=y_r)-p_r)]\\
    &[g(Z_2,t;p_r)+\mathbf{d}_1\mu(t;p_r)(I(Y_2=y_r)-p_r)]dM(t).
\end{align*}
Let $\{\lambda_i\}$ denote the eigenvalues of the operator $A$ defined by
\begin{align*}
    Aq(x)=\int h_*(z_1,z_2)q(y)d\nu(x_2)dP(y_2),
\end{align*}
then
\begin{align*}
nV_{n}^{(\hat{p}_r)}\stackrel{d}{\rightarrow}\sum_{j=1}^\infty\lambda_j^{(r)}\mathcal{X}_j^2(1), 
\end{align*}
where $\mathcal{X}_j^2(1)'s$, $j= 1,2,\ldots$, are independently and identically distributed chi-square distribution with 1 degree of freedom.

Notice that $\widehat{MDD}(X|Y)=\sum_{r=1}^R V_n^{(r)}$, according to the independence of the sample and the additivity of chi-square distribution,
\begin{equation*}
n\widehat{MDD}(X|Y)\stackrel{d}{\rightarrow}\sum_{j=1}^\infty\lambda_j\mathcal{X}_j^2(1),   
\end{equation*}
where $\lambda_j = \sum_{r=1}^R \lambda_j^{(r)}$.

% We have completed the decomposition form of Theorem 2.16 in \cite{de1987effect}, then only need to verify
% $\frac{\partial Eg(X,t;p_r)}{\partial p_r}=0$ under $H_0$.

% Because $Eg(X,t;p_r) = \sqrt{p_r}E[\frac{1}{p_r}\frac{1}{n}\sum_{k=1}^n I(x_k\in\bar{B}(X,t),Y_k=y_r)-\frac{1}{n}\sum_{k=1}^n I(x_k\in\bar{B}(X,t))] = 0$ under $H_0$, we have $\frac{\partial Eg(X,t;p_r)}{\partial p_r}=0$ under $H_0$, then we have $nJ_n=n\widehat{MDD}(X|Y)-nI_n\stackrel{p}{\rightarrow}0$.

% Above all, we combine $nI_n\stackrel{d}{\rightarrow}\sum_{j=1}^\infty\lambda_j\mathcal{X}_j^2(1)$ and $nJ_n\stackrel{p}{\rightarrow} 0$ as $n\rightarrow\infty$, we have $n\widehat{MDD}(X|Y) = nI_n+nJ_n \stackrel{d}{\rightarrow}\sum_{j=1}^\infty\lambda_j\mathcal{X}_j^2(1)$ as $n\rightarrow\infty$.
\end{proof}
%%%%%%%%%%%%%%%%%%%%%%%%%%%%%%%%%%

% We introduce Slutsky’s theorem \cite{Slutsky} first in order to prove Theorem \ref{the1.4}.

% \begin{lemma}[Slutsky’s theorem] \label{lm6.5}
% Let $X_{n},Y_{n}$ be sequences of scalar/vector/matrix random elements. If $X_{n}$ converges in distribution to a random element $X$ and $Y_{n}$ converges in probability to a constant $c$, then\\
% (i)$X_{n}+Y_{n}\stackrel{d}{\rightarrow} X+c$,
% (ii)$X_{n}Y_{n}\stackrel{d}{\rightarrow}Xc$,
% (iii)$X_{n}/Y_{n} \stackrel{d}{\rightarrow} X/c$, provided that c is invertible,\\
% where $\stackrel{d}{\rightarrow}$ denotes convergence in distribution.
% \end{lemma}

%Now, we try to prove Theorem \ref{the1.4} by Slutsky’s theorem.
\begin{proof}[\textbf{Proof of Theorem \ref{the1.4}}]
Under the alternative hypothesis $H_1$, $\widehat{MDD}(X|Y)\stackrel{a.s.}{\longrightarrow}MDD(X|Y) > 0$, as $n\rightarrow\infty$. Thus, we have $n\widehat{MDD}(X|Y)\stackrel{a.s.}{\longrightarrow}\infty$, as $n\rightarrow\infty$.
\end{proof}

%We use notations in the proof of the Theorem \ref{the1.3}.
\begin{proof}[\textbf{Proof of Theorem \ref{the1.5}}]
We consider
\begin{equation*}
\widehat{MDD}(X|Y) = I_n + J_n ,
\end{equation*}
where 
$$
\begin{aligned}
I_n=&\sum_{r=1}^R \frac{1}{n^4}\sum_{i,j,k,k^\prime=1}^n\frac{1}{p_r}I(X_k\in\bar{B}(X_i,X_j),Y_k=y_r)\\
&\quad\quad\quad\quad\quad\quad\quad\times
I(X_{k^\prime}\in\bar{B}(X_i,X_j),Y_{k^\prime}=y_r)\\
+&\sum_{r=1}^R \frac{1}{n^4}\sum_{i,j,k,k^\prime=1}^np_r I(X_k\in\bar{B}(X_i,X_j))I(X_{k^\prime}\in\bar{B}(X_i,X_j)) \\
-&\sum_{r=1}^R \frac{1}{n^4}\sum_{i,j,k,k^\prime=1}^n I(X_k\in\bar{B}(X_i,X_j),Y_k=y_r)\\
&\quad\quad\quad\quad\quad\quad\quad\times
I(X_{k^\prime}\in\bar{B}(X_i,X_j))\\
-&\sum_{r=1}^R \frac{1}{n^4}\sum_{i,j,k,k^\prime=1}^nI(X_{k^\prime}\in\bar{B}(X_i,X_j),Y_{k^\prime}=y_r)\\
&\quad\quad\quad\quad\quad\quad\quad\times
I(X_k\in\bar{B}(X_i,X_j))
\end{aligned}
$$
and
$$
\begin{aligned}
J_n=&\sum_{r=1}^R(\frac{1}{\hat{p}_r}-\frac{1}{p_r})\frac{1}{n^4}\sum_{i,j,k,k^\prime=1}^{n}I(X_k\in\bar{B}(X_i,X_j),Y_k=y_r)\\
&\quad\quad\quad\quad\quad\quad\quad\times
I(X_{k^\prime}\in\bar{B}(X_i,X_j),Y_{k^\prime}=y_r)\\
+& \sum_{r=1}^R(\hat{p}_r - p_r)\frac{1}{n^4}\sum_{i,j,k,k^\prime=1}^{n}I(X_k\in\bar{B}(X_i,X_j))\\
&\quad\quad\quad\quad\quad\quad\quad\times
I(X_{k^\prime}\in\bar{B}(X_i,X_j)).
\end{aligned}
$$

We consider $E[\Psi^{(r)}(z_i,Z_j,Z_k,Z_{k^\prime})]$ in the proof of Theorem \ref{the1.3},
\begin{equation*}
\begin{aligned}
&E[\Psi^{(r)}(z_i,Z_j,Z_k,Z_{k^\prime})]\\
=& \frac{1}{p_r}P_{j,k,k^\prime}[X_k,X_{k^\prime}\in\bar{B}(x_i,X_j),Y_k=Y_{k^\prime}=y_r)]\\
&+p_rP_{j,k,k^\prime}[X_k,X_{k^\prime}\in\bar{B}(x_i,X_j)]\\
&- P_{j,k,k^\prime}[X_k,X_{k^\prime}\in\bar{B}(x_i,X_j), Y_k=y_r]\\
&-P_{j,k,k^\prime}[X_k,X_{k^\prime}\in\bar{B}(x_i,X_j), Y_{k^\prime}=y_r].
\end{aligned}
\end{equation*}

Under the alternative hypothesis $H_1$, $X$ and $Y$ is not independent of each other, i.e.
\begin{equation*}
\begin{aligned}
&\frac{1}{p_r}P_{j,k,k^\prime}[X_k,X_{k^\prime}\in\bar{B}(x_i,X_j),Y_k=Y_{k^\prime}=y_r]\\
\neq &P_{j,k,k^\prime}[X_k,X_{k^\prime}\in\bar{B}(x_i,X_j),Y_k=y_r],
\end{aligned}
\end{equation*}
and 
\begin{equation*}
\begin{aligned}
&p_rP_{j,k,k^\prime}[X_k,X_{k^\prime}\in\bar{B}(x_i,X_j)]\\
\neq& P_{j,k,k^\prime}[X_k,X_{k^\prime}\in\bar{B}(x_i,X_j),Y_{k^\prime}=y_r],
\end{aligned}
\end{equation*}
so $E[\Psi^{(r)}(z_i,Z_j,Z_k,Z_{k^\prime})]$ is a non-constant function related to $z_i$, and we know $h_1^{(r)}(Z_1)=\frac{1}{4}[E[\Psi^{(r)}(z_i,Z_j,Z_k,Z_{k^\prime})]+E[\Psi^{(r)}(Z_i,z_j,Z_k,Z_{k^\prime})]+E[\Psi^{(r)}(Z_i,Z_j,z_k,Z_{k^\prime})]+E[\Psi^{(r)}(Z_i,Z_j,Z_k,z_{k^\prime})]]$, then we have $\Var[h_1^{(r)}(Z_1)]>0$.
We apply Lemma \ref{lm6.4} $(i)$, we have
\begin{equation*}
\sqrt{n}[V_n^{(p_r)}-E[\Psi^{(r)}(Z_i,Z_j,Z_k,Z_{k^\prime})]]\stackrel{d}{\rightarrow}N(0,16\Var[h_1^{(r)}(Z_1)]),
\end{equation*}

Because $I_n = \sum_{r=1}^R V_n^{(p_r)}$, $MDD(X|Y)=E[\sum_{r=1}^R p_r(F_r(X_i,X_j)-F(X_i,X_j))^2]$ and 
\begin{equation*}
\begin{aligned}
&E[\Psi^{(r)}(Z_i,Z_j,Z_k,Z_{k^\prime})]\\
=&E[E[\Psi^{(r)}(Z_i,Z_j,Z_k,Z_{k^\prime})|Z_i,Z_j]]\\
=&\frac{1}{p_r}E[E[I(X_k\in\bar{B}(X_i,X_j),Y_k=y_r)|Z_i,Z_j]\\
&\quad\quad
E[I(X_{k^\prime}\in\bar{B}(X_i,X_j),Y_{k^\prime}=y_r)|Z_i,Z_j]]\\
&+ p_r E[E[I(X_k\in\bar{B}(X_i,X_j))|Z_i,Z_j]\\
&\quad\quad
E[I(X_{k^\prime}\in\bar{B}(X_i,X_j))|Z_i,Z_j]]\\
&-E[E[I(X_k\in\bar{B}(X_i,X_j),Y_k=y_r)|Z_i,Z_j]\\
&\quad\quad\quad E[I(X_{k^\prime}\in\bar{B}(X_i,X_j))|Z_i,Z_j]] \\
&- E[E[I(X_{k^\prime}\in\bar{B}(X_i,X_j),Y_{k^\prime}=y_r)|Z_i,Z_j]\\
&\quad\quad E[I(X_k\in\bar{B}(X_i,X_j))|Z_i,Z_j]]\\
=&\frac{1}{p_r}E[E^2[I(X_k\in\bar{B}(X_i,X_j),Y_k=y_r)|Z_i,Z_j]] \\
&+ p_r E[E^2[I(X_k\in\bar{B}(X_i,X_j))|Z_i,Z_j]]\\
&-2E[E[I(X_k\in\bar{B}(X_i,X_j),Y_k=y_r)|Z_i,Z_j]\\
&\quad\quad\quad
E[I(X_k\in\bar{B}(X_i,X_j))|Z_i,Z_j]]\\
=&p_rE[(\frac{1}{p_r}E[I(X_k\in\bar{B}(X_i,X_j),Y_k=y_r)|Z_i,Z_j]\\
&-E[I(X_k\in\bar{B}(X_i,X_j))|Z_i,Z_j])^2]\\
=&E[p_r(F_r(X_i,X_j)-F(X_i,X_j))^2],
\end{aligned}   
\end{equation*}
we have
$$
E[\Psi^{(r)}(Z_i,Z_j,Z_k,Z_{k^\prime})]]=\sum_{r=1}^R MDD(X|Y)
$$
and
\begin{equation*}
\sqrt{n}[I_n-MDD(X|Y)]\stackrel{d}{\rightarrow}N(0,\sigma_I^2),
\end{equation*}
where $\sigma_I^2 = 
\sum_{r=1}^R16\Var[h_1^{(r)}(Z_i)]+2n\sum_{i<j}Cov(V_n^{(i)},V_n^{(j)})$. We explain here that $Cov(V_n^{(i)},V_n^{(j)})$ is the covariance of V-statistic $V_n^{(i)}$ and $V_n^{(j)}$ of order 4, which can be written as $Cov(V_n^{(i)},V_n^{(j)})=\frac{1}{n^4}\sum_{c=1}^4\binom{4}{c}(n-4)^{4-c}\sigma_{c,c}^2$, where $\sigma_{c,c}^2=Cov(h_c^{(p)}(Z_1,\ldots,Z_c),h_c^{(q)}(Z_1,\ldots,Z_c))$ and $h_c^{(p)}(Z_1,\ldots,Z_c)$ represents $h_c(Z_1,\ldots,Z_c)$ of Lemma \ref{lm6.4} when $r=p$.

Next, we would like to know the asymptotic distribution of $\sqrt{n}[\widehat{MDD}(X|Y)-MDD(X|Y)]$, we need to consider the asymptotic distribution of $\sqrt{n}J_n$ as follows
\begin{equation*}
\begin{aligned}
&\sqrt{n}[\widehat{MDD}(X|Y)-MDD(X|Y)] \\
=& \sqrt{n}[I_n+J_n-MDD(X|Y)]. 
\end{aligned}
\end{equation*}

Then, we consider $J_n$.
Denote $V_1^{(r)}=\frac{1}{n^4}\sum_{i,j,k,k^\prime=1}^{n}I(X_k\in\bar{B}(X_i,X_j),Y_k=y_r)I(X_{k^\prime}\in\bar{B}(X_i,X_j),Y_{k^\prime}=y_r)$ and $V_2=\frac{1}{n^4}\sum_{i,j,k,{k^\prime}=1}^{n}I(X_k\in\bar{B}(X_i,X_j))I(X_{k^\prime}\in\bar{B}(X_i,X_j))$, we have 
\begin{equation*}
J_n=\sum_{r=1}^R(\hat{p}_r-p_r)(V_2-\frac{1}{\hat{p}_rp_r}V_1^{(r)}).
\end{equation*}
The asymptotic distribution of $V_2-\frac{1}{\hat{p}_rp_r}V_1^{(r)}$ is a constant, because we know
$E[V_1^{(r)}]=E[I(X_k\in\bar{B}(X_i,X_j),Y_k=y_r)I(X_{k^\prime}\in\bar{B}(X_i,X_j),Y_{k^\prime}=y_r)]=p_r^2E[F_r^2(X_i,X_j)]$, $E[V_2]=E[I(X_k\in\bar{B}(X_i,X_j)I(X_{k^\prime}\in\bar{B}(X_i,X_j)]=E[F^2(X_i,X_j)]$ and $\hat{p}_r\rightarrow p_r$. Then 
$$
V_2-\frac{1}{\hat{p}_rp_r}V_1^{(r)}\rightarrow E[F^2(X_i,X_j)]-E[F_r^2(X_i,X_j)].
$$

According to the Central Limit Theorem (CLT), for arbitrary $r=1,\ldots,R$, we have  
$$
\sqrt{n}(\hat{p}_r-p_r)(V_2-\frac{1}{\hat{p}_rp_r}V_1^{(r)})\rightarrow N(0,\sigma_r^2),
$$
where $\sigma_r^2 = p_r(1-p_r)(E[F^2(X_i,X_j)]-E[F_r^2(X_i,X_j)])$.

Let $\mathbf{\hat{p}}^{(i)}=(I(Y_i=y_1),I(Y_i=y_2),\ldots,I(Y_i=y_R))^T$, where $\mathbf{\hat{p}}^{(i)}$ is a $R$-dimensional random variable and $\mathbf{\hat{p}}^{(1)}, \mathbf{\hat{p}}^{(2)},\ldots, \mathbf{\hat{p}}^{(n)}$ is dependent of each other. Therefore, according to multidimensional CLT, $\sqrt{n}(\frac{1}{n}\sum_{i=1}^n\mathbf{\hat{p}}^{(i)}-E[\mathbf{\hat{p}}^{(i)}])= (\sqrt{n}(\hat{p}_1-p_1),\ldots,\sqrt{n}(\hat{p}_R-p_R))^T$ asymptotically obey the $R$-dimensional normal distribution. In this way we get the condition of the additivity of the normal distribution, then we have
$$
\sqrt{n}J_n\rightarrow N(0,\sigma_J^2),
$$
where $\sigma_J^2 = \sum_{r=1}^R\sigma_r^2+2\sum_{i<j}Cov(\hat{p}_i,\hat{p}_j)= \sum_{r=1}^R\sigma_r^2-\frac{2}{n}\sum_{i<j}p_ip_j
=\sum_{r=1}^Rp_r(1-p_r)(E[F^2(X_i,X_j)]-E[F_r^2(X_i,X_j)])-\frac{2}{n}\sum_{i<j}p_ip_j$.

Similarly, we can use multidimensional CLT to prove $\sqrt{n}J_n$ and $\sqrt{n}[I_n-MV(X|Y)]$ are bivariate normal distribution. Then, we make a conclusion, 
\begin{equation*}
\sqrt{n}[\widehat{MDD}(X|Y)-MDD(X|Y)]\stackrel{d}{\rightarrow}N(0,\sigma^2),
\end{equation*}
where  $\sigma^2=\sigma_I^2+\sigma_J^2 + 2nCov(I_n,J_n)\sigma_I\sigma_J$. We explain here that $Cov(I_n,J_n)$ is the covariance of V-statistic $I_n$ of order 4 and $J_n$ of order 1, which can be written as $Cov(I_n,J_n)=Cov(\sum_{r=1}^R V_n^{(r)},\sum_{r=1}^R(\hat{p}_r-p_r)(V_2-\frac{1}{\hat{p}_rp_r}V_1^{(r)}))=\sum_{r=1}^R\sum_{r^{\prime}=1}^RCov(V_n^{(r)},(\hat{p}_{r^{\prime}}-p_{r^{\prime}})(V_2-\frac{1}{\hat{p}_{r^{\prime}}p_{r^{\prime}}}V_1^{(r)}))=\sum_{r=1}^R\sum_{r^{\prime}=1}^R\frac{1}{n}\binom{4}{1}(n-4)^3Cov(h_1^{(r)}(Z_1),I(Y_k=y_{r^{\prime}}))(E[F^2(X_i,X_j)]-E[F_r^2(X_i,X_j)])$.
\end{proof}
% For Original Research articles, please note that the Material and Methods section can be placed in any of the following ways: before Results, before Discussion or after Discussion.

%%  The bibliography

%%  If your bibliography is in BibTeX format, use the following setup:
%%  Style BST file for numbered citation:
%\bibliographystyle{imsart-number}
%%  or name-year citation
\bibliographystyle{imsart-nameyear}
%%  Bibliography file (usually `*.bib')
\bibliography{mdd_v1}          
%%
%%  or include bibliography directly:
%\begin{thebibliography}{9}
%%  Use \bibitem{r1} or \bibitem[Surname(2010)]{r1} (for authoryear case)
%%  Put author names in \textsc{} command in order to use small caps font
%
%\bibitem{}
%\textsc{}
%
%\end{thebibliography}

\end{document}